\documentclass[lettersize,journal]{IEEEtran}
\IEEEoverridecommandlockouts
\usepackage{amsfonts,amsthm,amssymb}
\usepackage{amsmath}
\usepackage{array}
\usepackage{textcomp}
\usepackage{subfigure}
\usepackage{float}
\usepackage{enumerate}
\usepackage{tabularx}
\usepackage{stfloats}
\usepackage{url}
\usepackage{hyperref}
\usepackage{xcolor,soul,framed} 
\usepackage{pstricks}
\usepackage{verbatim}
\usepackage{graphicx}
\usepackage{epstopdf}
\usepackage{booktabs}
\usepackage{cite}
\usepackage[T1]{fontenc}
\usepackage{tikz}
\usepackage{wrapfig}
\usepackage{multirow}
\usepackage{algorithm}
\usepackage{algorithmic}
\allowdisplaybreaks
\begin{document}
\newtheorem{proposition}{Proposition}
\newtheorem{lemma}{Lemma}
\renewcommand{\algorithmicrequire}{\textbf{Input:}} 
\renewcommand{\algorithmicensure}{\textbf{Output:}}
\newtheorem{Definition}{Definition}
\newtheorem{Assumption}{Assumption}
\newtheorem{Theorem}{Theorem}
\newtheorem{Corollary}{Corollary}
\newtheorem{Proposition}{Proposition}
\newtheorem{Lemma}{Lemma}
\newtheorem{Remark}{Remark}

\title{Optimal Configuration of Reconfigurable Intelligent Surfaces With Non-uniform Phase Quantization\\
\thanks{This work is supported by the National Natural Science Foundation of China (NSFC) under Grants NO.12141107 and NO.11801200, and the Interdisciplinary Research Program of HUST (2023JCYJ012), and the Guangxi Science and TechnologyProject (AB21196034).}
}

\author{Jialong~Lu,~\IEEEmembership{Student Member,~IEEE,}
        Rujing~Xiong,~\IEEEmembership{Graduate Student Member,~IEEE,}
        Tiebin Mi, Ke Yin,~\IEEEmembership{Member,~IEEE,} and
        Robert~Caiming~Qiu,~\IEEEmembership{Fellow,~IEEE}
\thanks{J. Lu, R.~Xiong, T. Mi and R. Qiu are with the School of Electronic Information and Communications, Huazhong University of Science and Technology, Wuhan 430074, China (e-mail: \{jialong, mitiebin, rujing, caiming\}@hust.edu.cn).}
\thanks{K. Yin is with the Center for Mathematical Sciences, Huazhong University of Science and Technology, Wuhan 430074, China (e-mail: kyin@hust.edu.cn).}
\thanks{Part of this paper was presented at the 2024 IEEE Global Communications
Conference (GLOBECOM 2024).}
}

\maketitle

\begin{abstract}
    The existing methods for reconfigurable intelligent surface (RIS) beamforming in wireless communications are typically limited to uniform phase quantization. However, in practical applications, engineering challenges and design requirements often lead to non-uniform phase and bit resolution of RIS units, which limits the performance potential of these methods. To address this issue, this paper pioneers the study of discrete non-uniform phase configuration in RIS-assisted multiple-input single-output (MISO) communication and formulates an optimization model to characterize the problem. For single-user scenarios, the paper proposes a partition-andtraversal (PAT) algorithm that efficiently achieves the global optimal solution through systematic search and traversal. For larger-scale multi-user scenarios, aiming to balance performance and computational complexity, an enhanced PAT-based algorithm (E-PAT) is developed. By optimizing the search strategy, the E-PAT algorithm significantly reduces computational overhead and achieves linear complexity. Numerical simulations confirm the effectiveness and superiority of the proposed PAT and EPAT algorithms. Additionally, we provide a detailed analysis of the impact of non-uniform phase quantization on system performance.
\end{abstract}

\begin{IEEEkeywords}
RIS, discrete phase configuration, global optimum, partition-and-traversa.
\end{IEEEkeywords}

\section{Introduction}

Reconfigurable Intelligent Surfaces (RISs) have garnered significant attention for their ability to reconfigure electronic environments\cite{xiong2024fair,di2020smart,wu2023intelligent,9982493}. Generally, a RIS consists of numerous low-power well-designed passive reflecting units, each capable of independently manipulating electromagnetic properties such as the phase of incident waves. This characteristic enables the RIS to redistribute incident waves, facilitating complex beamforming functionalities. Extensive research on RIS has been conducted in various wireless systems such as multi-user/multi-antenna systems~\cite{di2020hybrid,zheng2024intelligent}, unmanned aerial vehicle networks~\cite{li2020reconfigurable,bansal2023ris}, physical layer security~\cite{cui2019secure,pei2023secrecy}, wireless sensing and location~\cite{elzanaty2021reconfigurable,10284917,10143420}, and edge computing~\cite{bai2020latency,chen2022irs}, satellite communications~\cite{tekbiyik2021graph,jiang2024hybrid,kim2024performance} among others.

To harness the potential of RIS in wireless communications, appropriately optimizing its reflection such as the phase configuration is crucial. Assuming that a RIS consists of units with continuous phase shifts, phase optimization is not a challenging task and can be implemented through alternating direction method of multipliers~\cite{liang2016unimodular,yang2021beamforming}, successive convex approximation~\cite{wang2022sca,kumar2022novel}, semidefinite relaxation-semidefinite program (SDR-SDP)~\cite{elmossallamy2021ris,zhou2020robust}, majorization-minimization~\cite{shen2019secrecy,salem2022active} and manifold optimization (Manopt)~\cite{yu2019miso,hu2022constant}.

However, it is more practical to assume discrete phase shifts for RIS units. This limitation is due to the hardware structure of RIS units~\cite{wu2019beamforming}, making the continuous-phase-shift assumption unrealistic~\cite{xiong2023design,rana2023review}. With the discrete phase, traditional techniques are no longer applicable. One method to tackle this issue is to perform resource-intensive exponential search techniques~\cite{xiong2022optimal}. Given that a RIS can have hundreds and thousands of units, this method is extremely time-consuming. 

Additionally, researchers have put forward various techniques to obtain sub-optimal solutions. Closest point projection (CPP) is one of the most popular methods. This method hard rounding the continuous solution to its discrete counterpart~\cite{xiong2024optimal,wang2020intelligent}, which means it quantizes the solution derived from the continuous-phase-optimization problem~\cite{wu2021intelligent}. However, hard rounding may lead to performance degradation, and in the worst-case scenario, it can even result in arbitrarily poor performance~\cite{xiong2024optimal,zhang2022configuring}. Furthermore, the authors in~\cite{wu2019beamforming} investigated the communication optimization from a multi-user multi-antenna access point to multiple single-antenna users and proposed a successive refinement algorithm to optimize the RIS configuration with discrete phase shifts. The angle-of-arrival estimation similarity method was proposed in~\cite{xu2019discrete}, for designing finite-resolution discrete phase shifts. In~\cite{di2020hybrid,yu2020optimal}, the authors developed the branch-and-bound algorithms to enhance achievable rates under limited discrete phase shifts in RIS-aided multi-user communications. Moreover, a training-set-based approach was introduced in~\cite{an2021low}, which integrates multiple channel estimations and transmit precoding. For achieving an enhanced signal-to-noise ratio (SNR) boost in RIS-assisted communications, the authors in~\cite{zhang2022configuring} proposed an approximation algorithm to approach the global optimal solution, while the authors in~\cite{ren2022linear} presented a rotation-based algorithm that achieves the global optimum by considering single-input single-output systems.

All the aforementioned methods assume uniform phase shifts for the RIS unit, i.e., the discrete phase shifts are evenly spaced within the range $[0,2\pi)$. However, in practical scenarios, phase shifts are not necessarily uniform~\cite{rana2023review}. For example, the tunable units for phase adjustment on RIS primarily consist of positive-intrinsic-negative (PIN) diodes and varactor diodes. Studies~\cite{dai2019wireless,huang2017dynamical,zhang2019breaking,li2019machine,sayanskiy20222d,kamoda201160,chen2022transparent,tang2023transmissive,zhang2024design} reveal that, in these designs, regardless of whether the RIS is reflective or transmissive, practical discrete phase shift measurements are typically non-uniform, deviating from the simulation results due to engineering limitations. Additionally, there are also some non-uniform phase shifts deliberately introduced by human design~\cite{tekbıyık2022graph}. To clarify, we delineate non-uniform phase shifts into two types throughout this paper. The first type is discussed at the unit level and primarily stems from two sources: one is the inaccuracies occurring in the engineering process such as printed circuit board (PCB) manufacturing~\cite{araghi2022reconfigurable}. The other source arises from frequency variation, for instance, consider a RIS unit designed to provide uniform phase shifts at 3 GHz. If the same unit is used at 3.1 GHz, the phase shifts would not remain uniform~\cite{bjornson2021optimizing}. The second type of non-uniformity occurs at the array level, mainly due to different quantization schemes adopted by different RIS units, such as 1-bit and 2-bit quantization. In scenarios with varying rate requirements, adjusting the bit allocation allows for meeting demands while optimizing costs. To address unit-level non-uniformity phase configuration, reference~\cite{hashemi2024optimal} proposes an optimal configuration method capable of achieving a global optimum solution with linear complexity, although limited to single-input single-out (SISO) scenarios. 

To achieve the optimal phase configuration considering both unit-level and array-level non-uniformities, this paper formulates a discrete optimization problem aimed at minimizing the transmitted signal power while guaranteeing the received SNR in multi-input single-output (MISO) scenarios. The partition-and-traversal (PAT) algorithm and its evolutional version, the E-PAT, are proposed to address this problem: PAT achieves the global optimum~\footnote{This content has already been partly presented at IEEE GLOBECOM 2024.}, while E-PAT significantly reduces computational complexity while maintaining performance nearly identical to the global optimum.
\subsection{Contributions}
The main contributions of this paper can be summarized as follows:

\begin{itemize}
\item \textbf{Non-uniformity discrete phase configuration in RIS beamformings.} Considering the practical requirements and engineering challenges encountered by RIS hardware, we define two types of non-uniformity at the unit and array levels. Furthermore, to address the non-uniform discrete phase configuration problem in RIS-assisted MISO communications, we model the signal and introduce a constrained quadratic optimization problem aimed at minimizing the transmitted signal power while satisfying the minimum SNR constraint at the receiver.
\item \textbf{Novel algorithm to obtain the global optimum.} By introducing auxiliary variables, we reformulate an equivalent form of the original optimization problem. Subsequently, we propose a partition-and-traversal (PAT) algorithm to solve it. The key of the PAT algorithm is to partition the high-dimensional space using the introduced auxiliary variables and construct a feasible solution set. The global optimum of the original problem can then be achieved through exhaustive searching of this solution set.

\item \textbf{Efficient algorithm to reduce computational complexity while achieving comparable performance to the global optimum.} Building upon the PAT algorithm, we propose an improved, efficient search algorithm that significantly reduces the computational complexity of the PAT algorithm while maintaining comparable transmit power performance. We performed theoretical analysis and experimental validation to compare the complexities and practical performances of both proposed algorithms. Simulation results demonstrate that both of the proposed algorithms outperform the benchmarks. Moreover, we emphasize that proposed algorithm has the potential to achieve a high probability of finding the global optimal solution with a computational complexity of $O(N)$.
\end{itemize}

\subsection{Organization}

The remainder of the paper is organized as follows. In Section~\ref{Section2}, we present the modeling of multi-RIS-assisted multi-user communications and formulate the beamforming problem. Section~\ref{Section3} introduces the PAT algorithm, which efficiently addresses the problem and guarantees a global optimal solution. In Section~\ref{Section4}, we further analyze and propose an efficient extension of the PAT algorithm, referred to as the E-PAT algorithm. Section~\ref{Section5} is dedicated to the performance evaluations of the proposed algorithm through numerical tests. Finally, the paper is concluded in Section~\ref{Section6}.

\subsection{Notations}

The imaginary unit is denoted by $j$. The magnitude, real and complex components of a complex number are represented by $|\cdot|$, $\Re(\cdot)$ and $\Im(\cdot)$, respectively. Unless explicitly specified, lower and upper case bold letters denote vectors and matrices. The conjugate transpose, conjugate, and transpose of $\mathbf{A}$ are written as $\mathbf{A}^\mathrm{H}$, $\mathbf{A}^*$ and $\mathbf{A}^\mathrm{T}$, respectively. $\mathrm{diag}(\cdot)$ represents the operation of generating a diagonal matrix. 

\section{System Model And Problem Formulation}\label{Section2}

\subsection{ System Model}

As illustrated in Fig~\ref{1}, We consider a downlink communication system assisted by multiple RISs within a single-cell network. In this scenario, different types of RIS are strategically deployed to facilitate communication from a multi-antenna access point (AP) to $M$ single-antenna users, operating within a designated frequency band. Let $D$ be the number of transmit antennas at the AP. Suppose the system is assisted by $K$ RISs, where the $k$-th RIS consists of $n_k$ reflecting units. The total number of RIS units is given by $\sum_{k=1}^{K}n_k=N$. The AP employs a pre-designed active beamforming vector $\mathbf{w}$ for precoding, and the transmitted signal $s$ follows a zero-mean, unit-variance distribution. Due to severe path loss, multi-bounce reflections at the RIS are assumed to be negligible and are thus omitted from the analysis. The reflection link is akin to the dyadic backscatter channel in radio frequency identification communications~\cite{boyer2013invited}, specifically where the signal is emitted from the RIS  as a point source signal to the user, with an additional phase attached. Consequently, the equivalent baseband channel for the reflected link from RIS $k$ to user $m$ can be decomposed into the AP-RIS link, RIS reflection phase shift, and RIS-user link, denoted by $\mathbf{G}_{k}\in \mathbb{C}^{n_k\times D}$, $\mathbf{\Omega}_{k}=\mathrm{diag}(e^{j\Omega_{1}^k},e^{j\Omega_2^k},\cdots,e^{j\Omega_{n_k}^k}) \in \mathbb{C}^{n_k\times n_k}$ and $\mathbf{h}_{k,m}^\mathrm{H} \in \mathbb{C}^{1\times n_k}$. 
 \begin{figure}[htbp]
    \centering
    \includegraphics[width=3.4in]{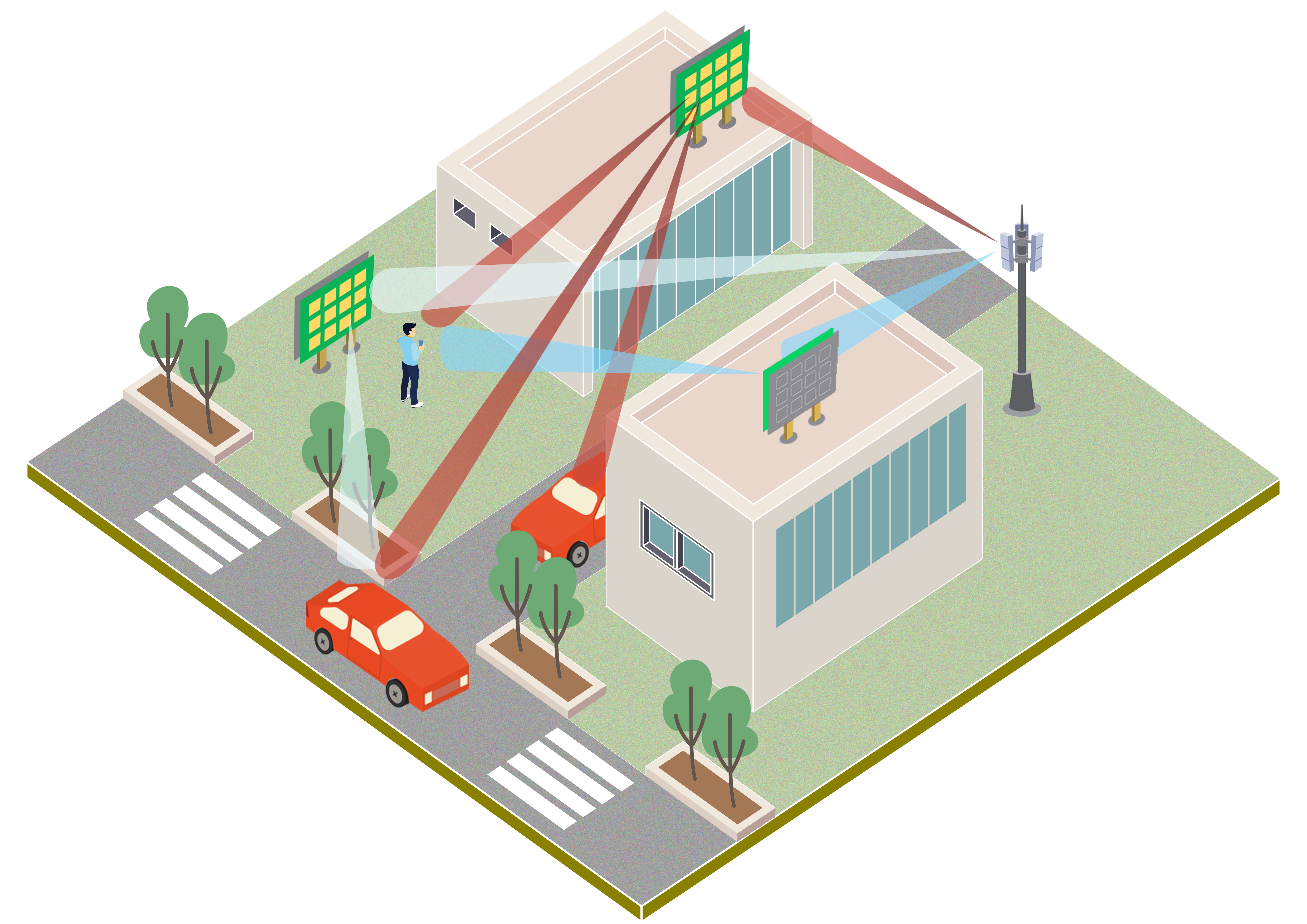}
    \caption{multi-RIS-assisted multi-user communication.}
    \label{1}
  \end{figure}

 As previously studied in \cite{10260279}, when the RIS $k$ is illuminated by the AP in the far field, as shown in Fig~\ref{2}, the signal $y_{k,m}$ received by the user $m$ from RIS $k$ can be expressed as given in (1), where $\mathbf{u}(\theta,\phi)=[\sin\theta \cos\phi, \sin\theta \sin\phi, \cos\theta]^\mathrm{T}$, $(r_{k,m}^s, \theta_{k,m}^s, \phi_{k,m}^s)$ and $(r_{k,d}^i, \theta_{k,d}^i, \phi_{k,d}^i)$ denote the polar coordinates of the user $m$ and the $d$-th antenna of the AP relative to RIS $k$, respectively. $\mathbf{p}_{k,n}$ represents the position of the $n$-th unit of RIS $k$, $\xi$ denotes the scattering pattern of the isotropic , and $\lambda$ is the wavelength of the electromagnetic wave. Therefore, the received signal at the user $m$ can be expressed as
 \begin{figure*}[!htbp]
    \begin{equation}\label{E:MIMO}
      \begin{aligned}
          y_{k,m} 
        = &  \overbrace{
        \xi \frac{ e^{-j 2 \pi r_{k,m}^s / \lambda }}{ r_{k,m}^s } 
        \begin{bmatrix}
          e^{ j 2 \pi \mathbf{p}_{k,1}^{\mathrm{T}} \mathbf{u} (\theta_{k,m}^s, \phi_{k,m}^s)  / \lambda } & \cdots & e^{ j 2 \pi \mathbf{p}_{k,n_k}^{\mathrm{T}} \mathbf{u} (\theta_{k,m}^s, \phi_{k,m}^s)  / \lambda } 
        \end{bmatrix} }^{\mathbf{h}_{k,m}^\mathrm{H}}
        \overbrace{\begin{bmatrix}
          e^{j \Omega_{1}^k} &        &  0 \\
                                 & \ddots &    \\
              0                  &        & e^{j \Omega_{n_k}^k}
        \end{bmatrix}}^{\mathbf{\Omega}_k} \\
        & \underbrace{\begin{bmatrix}
          e^{ j 2 \pi \mathbf{p}_{k,1}^{\mathrm{T}} \mathbf{u} (\theta_{k,1}^i, \phi_{k,1}^i)  / \lambda }  & \cdots  & e^{ j 2 \pi \mathbf{p}_{k,1}^{\mathrm{T}} \mathbf{u} (\theta_{k,D}^i, \phi_{k,D}^i)  / \lambda } \\
          \vdots  & \ddots & \vdots \\
          e^{ j 2 \pi \mathbf{p}_{k,n_k}^{\mathrm{T}} \mathbf{u} (\theta_{k,1}^i, \phi_{k,1}^i)  / \lambda }  & \cdots  & e^{ j 2 \pi \mathbf{p}_{k,n_k}^{\mathrm{T}} \mathbf{u} (\theta_{k,D}^i, \phi_{k,D}^i)  / \lambda }\\
        \end{bmatrix}
        \begin{bmatrix}
          \frac{ e^{-j 2 \pi r_{k,1}^i / \lambda } }{ r_{k,1}^i } &        &  0 \\
                                 & \ddots &    \\
              0                  &        &  \frac{ e^{-j 2 \pi r_{k,D}^i / \lambda } }{ r_{k,D}^i }
        \end{bmatrix}}_{\mathbf{G}_k}
       \mathbf{w}s
      \end{aligned}
    \end{equation}
    \medskip
    \hrule
\end{figure*}
 \begin{equation}
     y_m = (\mathbf{g}_m^\mathrm{H}+\sum_{k=1}^{K}\mathbf{h}_{k,m}^\mathrm{H} \mathbf{\Omega}_{k} \mathbf{G}_{k})\mathbf{w}s + z_m,
 \end{equation}
where $z\sim \mathcal{CN}(0,\sigma_m^2)$ denotes the additive white Gaussian noise (AWGN) at the user $m$ and $\mathbf{g}_m^\mathrm{H} \in \mathbb{C}^{1\times D}$ represents the equivalent baseband model of the direct link from the AP to user $m$. Correspondingly, the Signal-to-Noise Ratio (SNR) at the user $m$ is given by
\begin{equation}
    \mathrm{SNR}_m = \frac{|(\mathbf{g}_m^\mathrm{H}+\sum_{k=1}^{K}\mathbf{h}_{k,m}^\mathrm{H} \mathbf{\Omega}_{k} \mathbf{G}_{k})\mathbf{w}|^2}{\sigma_m^2}.
    \label{e2}
\end{equation}
\begin{figure}[htbp]
    \centering
    \includegraphics[width=3.4in]{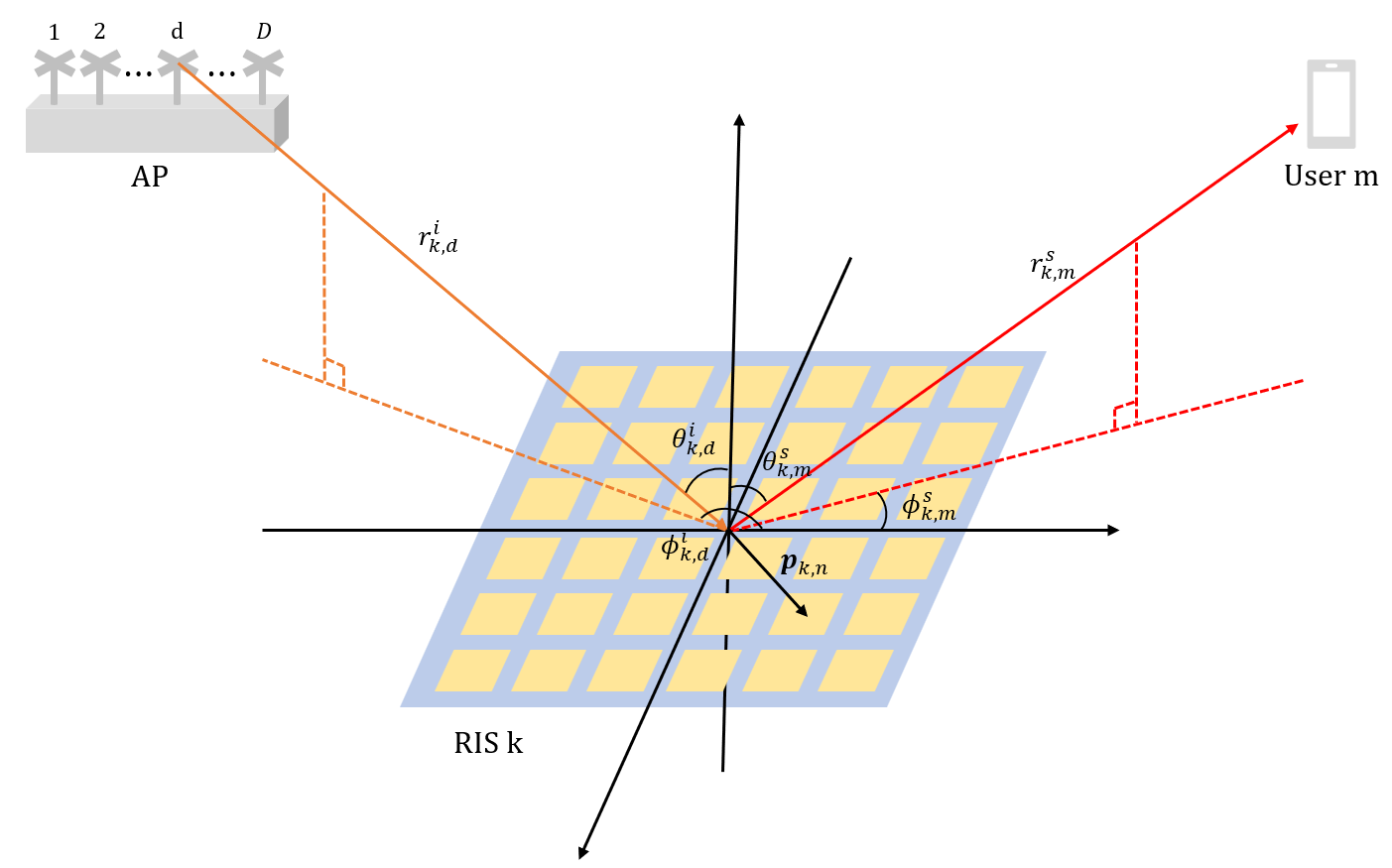}
    \caption{schematic diagram of the radiation model}
    \label{2}
  \end{figure}

\subsection{Beamforming Problem Formulation}

Our objective is to jointly optimize the beamforming of the AP and the RIS, aiming to minimize the AP’s transmit power while ensuring that the user's average SNR meets a certain threshold. Specifically, we seek to solve the following problem
\begin{equation}
    \begin{aligned}
        (\mathrm{P1}): \ &\min_{\mathbf{w},\mathbf{\Omega}_1,\mathbf{\Omega}_2,\cdots,\mathbf{\Omega}_K}||\mathbf{w}||^2\\
        &\mathrm{s.t.} \ \frac{1}{M}\sum_{m=1}^M\mathrm{SNR}_m \ge \gamma,\\
    \end{aligned}
\end{equation}
where $\gamma$ denotes the lower bound of the user's average SNR. We define the augmented vector $\tilde{\mathbf{v}}=[\mathbf{v}^\mathrm{T},1]^\mathrm{T}$, and channel matrix $\tilde{\mathbf{R}}_m^\mathrm{H}=[\mathbf{R}_m^\mathrm{H};\mathbf{g}_m^\mathrm{H}]$, where $\mathbf{v}=[e^{j\theta_1},e^{j\theta_2},\cdots,e^{j\theta_N}]^\mathrm{T}=[e^{j\Omega_1^1},\cdots,e^{j\Omega_{n_1}^1},e^{j\Omega_1^2},\cdots,e^{j\Omega_{n_K}^K}]^\mathrm{T}\in\mathbb{C}^{N\times 1}$ and $\mathbf{R}_m^\mathrm{H} = \mathrm{diag}([\mathbf{h}_{1,m}^\mathrm{H},\mathbf{h}_{2,m}^\mathrm{H},\cdots,\mathbf{h}_{K,m}^\mathrm{H}])\cdot[\mathbf{G}_1^\mathrm{T},\mathbf{G}_2^\mathrm{T},\cdots,\mathbf{G}_K^\mathrm{T}]^\mathrm{T}\in\mathbb{C}^{N\times D}$. Therefore, we can express $|(\mathbf{g}_m^\mathrm{H}+\sum_{k=1}^{K}\mathbf{h}_{k,m}^\mathrm{H} \mathbf{\Omega}_{k} \mathbf{G}_{k})\mathbf{w}|^2=|\tilde{\mathbf{v}}^\mathrm{H}\tilde{\mathbf{R}}_m^\mathrm{H}\mathbf{w}|^2$, and the problem (P1) can be reformulated as
\begin{equation}
    \begin{aligned}
        (\mathrm{P2}): \ &\min_{\mathbf{w},\mathbf{v}}||\mathbf{w}||^2\\
        &\mathrm{s.t.} \ \frac{1}{M}\sum_{m=1}^M\frac{|\tilde{\mathbf{v}}^\mathrm{H}\tilde{\mathbf{R}}_m^\mathrm{H}\mathbf{w}|^2}{\sigma_m^2} \ge \gamma.\\
    \end{aligned}
\end{equation}
Given $\mathbf{v}$, the Lagrangian function associated with problem (P2) is expressed as:
\begin{equation}
    \mathcal{L}(\mathbf{w},\eta) = ||\mathbf{w}||^2 + \eta (M\gamma-\sum_{m=1}^M\frac{|\tilde{\mathbf{v}}^\mathrm{H}\tilde{\mathbf{R}}_m^\mathrm{H}\mathbf{w}|^2}{\sigma_m^2}),
\end{equation}
where $\eta$ is the Lagrange multiplier. Taking the partial derivative of $\mathcal{L}(\mathbf{w},\eta)$ with respect to $\mathbf{w}$ and setting it to zero yields:
\begin{equation}
    \mathbf{w}^* = \frac{\sqrt{P}\mathbf{q}}{\sigma_1^2\sigma_2^2\cdots\sigma_M^2}, \ \eta^* = \frac{\sigma_1^2\sigma_2^2\cdots\sigma_M^2}{\mu},
\end{equation}
where $\mathbf{q}$ is the unit eigenvector of the matrix $\sum_{m=1}^{M}(\prod_{i \ne m}^{M}\sigma_i^2)\tilde{\mathbf{R}}_m\tilde{\mathbf{v}}\tilde{\mathbf{v}}^\mathrm{H}\tilde{\mathbf{R}}_m^\mathrm{H}$, and $\mu$ is the corresponding eigenvalue. By substituting equation (7) into (P2), we transform problem (P2) into
\begin{equation}
    \begin{aligned}
        (\mathrm{P3}): \ &\min_{\mathbf{w},\mathbf{v}}P\\
        &\mathrm{s.t.} \ \frac{\mu P}{M \sigma_1^2\sigma_2^2\cdots\sigma_M^2} \ge \gamma.\\
    \end{aligned}
\end{equation}
By observation, the optimal solution to problem (P3) is obtained when $P = \frac{\gamma M \sigma_1^2\sigma_2^2\cdots\sigma_M^2}{\mu_{max}}$. Therefore, we can reformulate the problem as
\begin{equation}
    \begin{aligned}
        &\max_{\mathbf{v}}\mu_{max}(\sum_{m=1}^{M}(\prod_{i \ne m}^{M}\sigma_i^2)\tilde{\mathbf{R}}_m\tilde{\mathbf{v}}\tilde{\mathbf{v}}^\mathrm{H}\tilde{\mathbf{R}}_m^\mathrm{H})
    \end{aligned}
\end{equation}

In practical applications, the additional phase imposed by the reflection units of the RIS is generally limited to discrete values~\cite{xiong2023design}. We define the set $\Phi_n = \{\phi_1^{n},\phi_2^{n},\cdots,\phi_{b_n}^{n}\}$ to represent the optional phase set of $\theta_n$. Here, $b_n$ denotes the number of optional phase values for $\theta_n$, and $\phi_1^{n},\phi_2^{n},\cdots,\phi_{b_n}^{n}$ are distributed in ascending order within the range $[0,2\pi)$. 

Optimization problem (9) considers the case where a direct link exists between the BS and UE. In fact, as demonstrated in reference~\cite{xiong2022optimal}, the optimization problem formulation in the absence of a direct link remains the same, with the only difference being the length of the variables. Therefore, we focus on the scenario without a direct link and solve the simplified problem, as
\begin{equation}
    \begin{aligned}
        (\mathrm{P4}): \ &\max_{\mathbf{v}}\mu_{max}(\sum_{m=1}^{M}(\prod_{i \ne m}^{M}\sigma_i^2)\mathbf{R}_m\mathbf{v}\mathbf{v}^\mathrm{H}\mathbf{R}_m^\mathrm{H})\\
        &\mathrm{s.t. \ arg}\{v_n\}\in \Phi_n, \forall{n}.
    \end{aligned} 
\end{equation} 

In conventional phase configuration optimization, it is typically assumed that $\phi_1^{n},\phi_2^{n},\cdots,\phi_{b_n}^{n}$ are uniformly distributed within $[0,2\pi)$, and that $b_n$ remains consistent across all elements, an assumption widely adopted in previous studies. However, in practical scenarios, due to the manufacturing precision limitations of RIS hardware and the necessity to accommodate RIS elements with different bit resolutions, $\phi_1^{n},\phi_2^{n},\cdots,\phi_{b_n}^{n}$ may exhibit non-uniform distributions within $[0,2\pi)$, and $b_n$ may also vary. 

Under such non-uniform quantization conditions, the objective function exhibits highly complex nonlinear characteristics within the solution space, with its surface presenting multiple local optima and significant gradient variations between these points. This results in substantial fluctuations in the function values across different regions. The non-convexity and non-smoothness of the function make traditional optimization algorithms prone to becoming trapped in local optima, thereby hindering convergence to the global optimum. Furthermore, the high volatility of the function surface significantly increases the complexity of the search space, leading to slower convergence rates and substantially higher computational costs for global optimization algorithms. Consequently, there is an urgent need to develop novel optimization methods capable of effectively addressing the challenges posed by non-uniform quantization. Such methods should aim to enhance optimization performance, improve the quality of global solutions, and ensure computational efficiency to meet the demands of practical systems.

\section{Global optimum for single-user scenarios}\label{Section3}

In this section, we further analyze the problem (P4) through transformation. For the single-user scenario, the problem can be equivalently reformulated as a constrained quadratic function maximization problem. To address this, we propose a PAT-based algorithm to obtain the global optimal solution efficiently. Specifically, by introducing auxiliary variables to construct mapping relationships between variables, we reduce the original optimization problem from a high-dimensional space to a low-dimensional space for solving. With this mapping relationship, we can partition the space where auxiliary variables reside into several subregions, with each subregion corresponding to a solution. Finally, by traversing these subregions, we obtain the global optimal solution.

\subsection{Simplification Through Problem Reformulation}

When $M=1$, we can simplify problem (P4) to
\begin{equation}
    \begin{aligned}
        &\max_{\mathbf{v}}\mu_{max}(\mathbf{R}\mathbf{v}\mathbf{v}^\mathrm{H}\mathbf{R}^\mathrm{H})\\
        &\mathrm{s.t. \ arg}\{v_n\}\in \Phi_n, \forall{n}.
    \end{aligned}
\end{equation} 
Since $\mathbf{R}\mathbf{v} \in \mathbb{C}^{D \times 1}$, $\mathbf{R}\mathbf{v}\mathbf{v}^\mathrm{H}\mathbf{R}^\mathrm{H}$ has only one non-zero eigenvalue, which is $\mathbf{v}^\mathrm{H}\mathbf{R}^\mathrm{H}\mathbf{R}\mathbf{v}$. Therefore, the problem is equivalent to
\begin{equation}
    \begin{aligned}
        (\mathrm{P5}): \ &\max_{\mathbf{v}}\mathbf{v}^\mathrm{H}\mathbf{R}^\mathrm{H}\mathbf{R}\mathbf{v}\\
        &\mathrm{s.t. \ arg}\{v_n\}\in \Phi_n, \forall{n}.
    \end{aligned}
\end{equation} 

(P5) falls into the category of common positive semidefinite quadratic maximization problems encountered in communication system design. Although existing studies~\cite{AVIS199621,doi:10.1137/0215024,FERREZ200535} have proposed various algorithms to solve such problems, their high computational complexity limits their feasibility in practical applications. Meanwhile, low-complexity algorithms such as CPP significantly reduce computational overhead but often suffer from unpredictable performance degradation due to their reliance on hard rounding. Moreover, the phase non-uniformity introduces irregularities in the problem's feasible region, further complicating the problem and exacerbating the losses due to discretization. To overcome these challenges, we propose a discrete global optimization algorithm in the next section.
\subsection{A Partition-and-Traversal Algorithm  For Positive Semidefinite Quadratic Maximization}

\subsubsection{An Equivalent Formulation for Positive Semidefinite Quadratic Maximization}

We know that $\mathbf{v}^\mathrm{H}\mathbf{R}^\mathrm{H}\mathbf{Rv}=||\mathbf{Rv}||^2$. By introducing the auxiliary variable $\overline{\mathbf{v}}=\begin{bmatrix} e^{j\varphi_1}\sin\vartheta_1 \\ e^{j\varphi_2}\cos\vartheta_1\sin\vartheta_2 \\ \vdots \\ e^{j\varphi_D}\cos\vartheta_1\cos\vartheta_2 \ldots \cos\vartheta_{D-1} \end{bmatrix}\in \mathbb{C}^{D\times1}$, and combining it with the Cauchy-Schwarz inequality, the equation finds 
\begin{equation}
    \Re\{\overline{\mathbf{v}}^\mathrm{H} \mathbf{Rv}\} \le ||\overline{\mathbf{v}}|| \cdot ||\mathbf{Rv}|| = ||\mathbf{Rv}||.
\end{equation} 
Therefore, (P5) is equivalent to
\begin{equation}
    \begin{aligned}\mathrm{(P6)}: \ &\max_{\mathbf{v}}\max_{\overline{\mathbf{v}}}\Re\{\overline{\mathbf{v}}^\mathrm{H} \mathbf{Rv}\}\\
        &\mathrm{s.t. \ arg}\{v_n\}\in \Phi_n, \forall{n},\\
        & \ \ \ \ ||\overline{\mathbf{v}}|| = 1.
    \end{aligned}
    \label{e8}
\end{equation} 
\vspace{0.1in}
Let $\mathbf{a}=\overline{\mathbf{v}}^\mathrm{H} \mathbf{R} \in \mathbb{C}^{1\times N}$, and $a_n = |a_n|e^{-j\tau_n}$ represent the $n$-th element in $\mathbf{a}$. Then (P6) can be rewritten as:
\begin{equation}
    \begin{aligned}
        &\max_{\mathbf{v}}\max_{\overline{\mathbf{v}}}\Re\{\overline{\mathbf{v}}^\mathrm{H} \mathbf{Rv}\}\\
        =&\max_{\overline{\mathbf{v}}}\max_{\mathbf{v}}\Re\{\overline{\mathbf{v}}^\mathrm{H} \mathbf{Rv}\}\\
        =&\max_{\overline{\mathbf{v}}}\max_{\mathbf{v}}\sum_{n=1}^N\Re \{|a_n|e^{j(\theta_n-\tau_n)}\}\\ 
        =&\max_{\overline{\mathbf{v}}}\max_{\mathbf{v}}\sum_{n=1}^N|a_n|\cos(\theta_n-\tau_n).
    \end{aligned}
    \label{e9}
\end{equation}

\subsubsection{Subregion Partitioning}

From (15), it can be observed that when $\overline{\mathbf{v}}$ is determined, the condition for maximizing $\Re\{\overline{\mathbf{v}}^\mathrm{H} \mathbf{Rv}\}$ is given by
\begin{equation}
    \theta_n(\overline{\mathbf{v}})=\arg \min_{\theta_n \in \Phi_n}|(\theta_n-\tau_n)\mod 2\pi|.
    \label{e10}
\end{equation}
We define the set $\Psi_n=\{\psi_1^n,\psi_2^n,\cdots,\psi_{b_n}^n\}$, with 
\begin{equation}
    \psi_i^n=\phi_i^n+\frac{1}{2}[(\phi_{i+1}^n-\phi_i^n)\mod 2\pi],
    \label{e11}
\end{equation}
where the index $i$ is interpreted modulo $b_n$. Then, the maximization condition (16) can be rewritten as
\begin{equation}
    \theta_n(\overline{\mathbf{v}})=\phi_i^n, \ \ \mathrm{if} \ \tau_n \in (\psi_{i-1}^n,\psi_i^n)
    \label{e12}
\end{equation}
Based on condition (18), we can establish a mapping relationship between $\overline{\mathbf{v}}$ and $\mathbf{v}$. Then, let $\tau_n$ take the boundary value from equation (18), we can establish boundary surface equations to partition the high-dimensional space where $\overline{\mathbf{v}}$ resides into several subregions, each corresponding to a unique $\mathbf{v}$, as illustrated in Fig~\ref{3}. As there must exist a subregion corresponding to the global optimal solution of (P6), traversing all subregions allows us to obtain the global optimal solution $\mathbf{v}^*$. Since the dimensionality of $\overline{\mathbf{v}}$ is lower than that of $\mathbf{v}$, the number of subregions we partition is significantly lower than the original spatial scale of $\mathbf{v}$. This characteristic leads to a substantial reduction in computational complexity.

The boundary surface equation can be expressed as:
\begin{equation}
    \mathbf{R}_n^\mathrm{H}\overline{\mathbf{v}}=ae^{j\omega_n},
    \label{e13}
\end{equation}
where $\mathbf{R}_n^\mathrm{H}$ denotes the $n$-th row of $\mathbf{R}^\mathrm{H}$, $a$ is any positive real number, and $\omega_n\in\Psi_n$.
\begin{figure}
  \centering
  \includegraphics[width=3.4in]{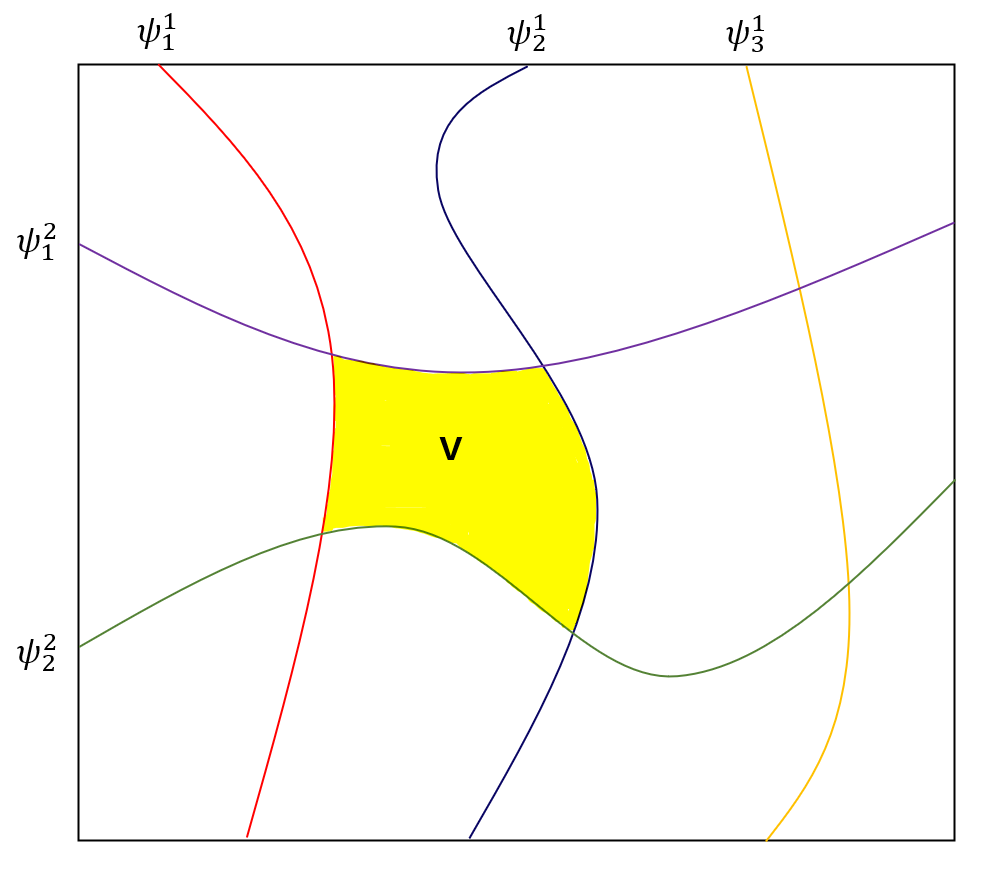}
  \caption{The spatial partitioning of $\overline{\mathbf{v}}$ results in each subregion corresponding to a distinct $\mathbf{v}$.}
 \label{3}
\end{figure}

\subsubsection{Determination and Traversal of subregions }

After partitioning the subregions, we employ the method of finding intersection points to differentiate and determine the subregions. Since $\overline{\mathbf{v}}$ has $2D-1$ variables, we require $2D-1$ boundary surface equations to determine an intersection point. We define the set of indices for the selected $2D-1$ equations as $I=\{i_1,i_2,\cdots,i_{2D-1}\}\subset\{1,2,\cdots,2D-1\}$, with corresponding boundary values set as $\Omega=\{\omega_{i_1},\omega_{i_2},\cdots,\omega_{i_{2D-1}}\}$. Then, the equation for the intersection point is
\begin{equation}
    \begin{bmatrix} \mathbf{R}_{i_1}^\mathrm{H} \\ \mathbf{R}_{i_2}^\mathrm{H} \\ \vdots \\ \mathbf{R}_{i_{2D-1}}^\mathrm{H} \end{bmatrix}\overline{\mathbf{v}}=\begin{bmatrix} a_1e^{j\omega_{i_1}} \\ a_2e^{j\omega_{i_2}} \\ \vdots \\ a_{2D-1}e^{j\omega_{i_{2D-1}}} \end{bmatrix}.
    \label{e14}
\end{equation}
Through fractional simplification, we obtain
\begin{equation}
    \begin{bmatrix} \mathbf{R}_{i_1}^\mathrm{H}e^{-j\omega_{i_1}} \\ \mathbf{R}_{i_2}^\mathrm{H}e^{-j\omega_{i_2}} \\ \vdots \\ \mathbf{R}_{i_{2D-1}}^\mathrm{H}e^{-j\omega_{i_{2D-1}}} \end{bmatrix}\overline{\mathbf{v}}=\begin{bmatrix} a_1 \\ a_2 \\ \vdots \\ a_{2D-1} \end{bmatrix}.
    \label{e15}
\end{equation}
Let $\mathbf{C}=\begin{bmatrix} \mathbf{R}_{i_1}^\mathrm{H}e^{-j\omega_{i_1}} \\ \mathbf{R}_{i_2}^\mathrm{H}e^{-j\omega_{i_2}} \\ \vdots \\ \mathbf{R}_{i_{2D-1}}^\mathrm{H}e^{-j\omega_{i_{2D-1}}} \end{bmatrix}$. Thus, equation (21) is equivalent to
\begin{subequations}
    \begin{align}
        \begin{bmatrix} \Re\{\mathbf{C}\} & \Im\{\mathbf{C}\} \end{bmatrix} \begin{bmatrix} \Im\{\overline{\mathbf{v}}\} \\ \Re\{\overline{\mathbf{v}}\} \end{bmatrix}=0 \\
        \begin{bmatrix} \Re\{\mathbf{C}\} & -\Im\{\mathbf{C}\} \end{bmatrix} \begin{bmatrix} \Re\{\overline{\mathbf{v}}\} \\ \Im\{\overline{\mathbf{v}}\} \end{bmatrix}>0
    \end{align}
\end{subequations}

By solving equation (22), we can determine an intersection point $\overline{\mathbf{v}}$ and compute the corresponding $\mathbf{v}$ based on equation (18), thereby determining and distinguishing subregions. Let $V$ denote the set of recovered $\mathbf{v}$. It should be noted that the recovery of $\theta_n$ for $n\in I$ is uncertain due to $\tau_n$ being on the boundary. To address this issue, we employ a traversal method. Assuming $\tau_n=\psi_i^n$, we include both cases where $\theta_n$ equals $\phi_i^n$ or $\phi_{i+1}^n$ in $V$. Therefore, each intersection point $\overline{\mathbf{v}}$ will yield $2^{2D-1} \ \mathbf{v}$. Interestingly, the set of $2^{2D-1} \ \mathbf{v}$ obtained actually corresponds to the $\mathbf{v}$ associated with all subregions connected to $\overline{\mathbf{v}}$. As shown in Fig~\ref{4}, by identifying one intersection point $\overline{\mathbf{v}}$, we can determine $\mathbf{v}_1, \mathbf{v}_2, \mathbf{v}_3$ and $\mathbf{v}_4$, which correspond to the subregions connected to the intersection point. Thus, by identifying all intersection points, we can determine all subregions and subsequently search for the global optimal solution. Finally, traversing $V$ yields the global optimal solution $\mathbf{v}^*$.
\begin{figure}
  \centering
  \includegraphics[width=3.4in]{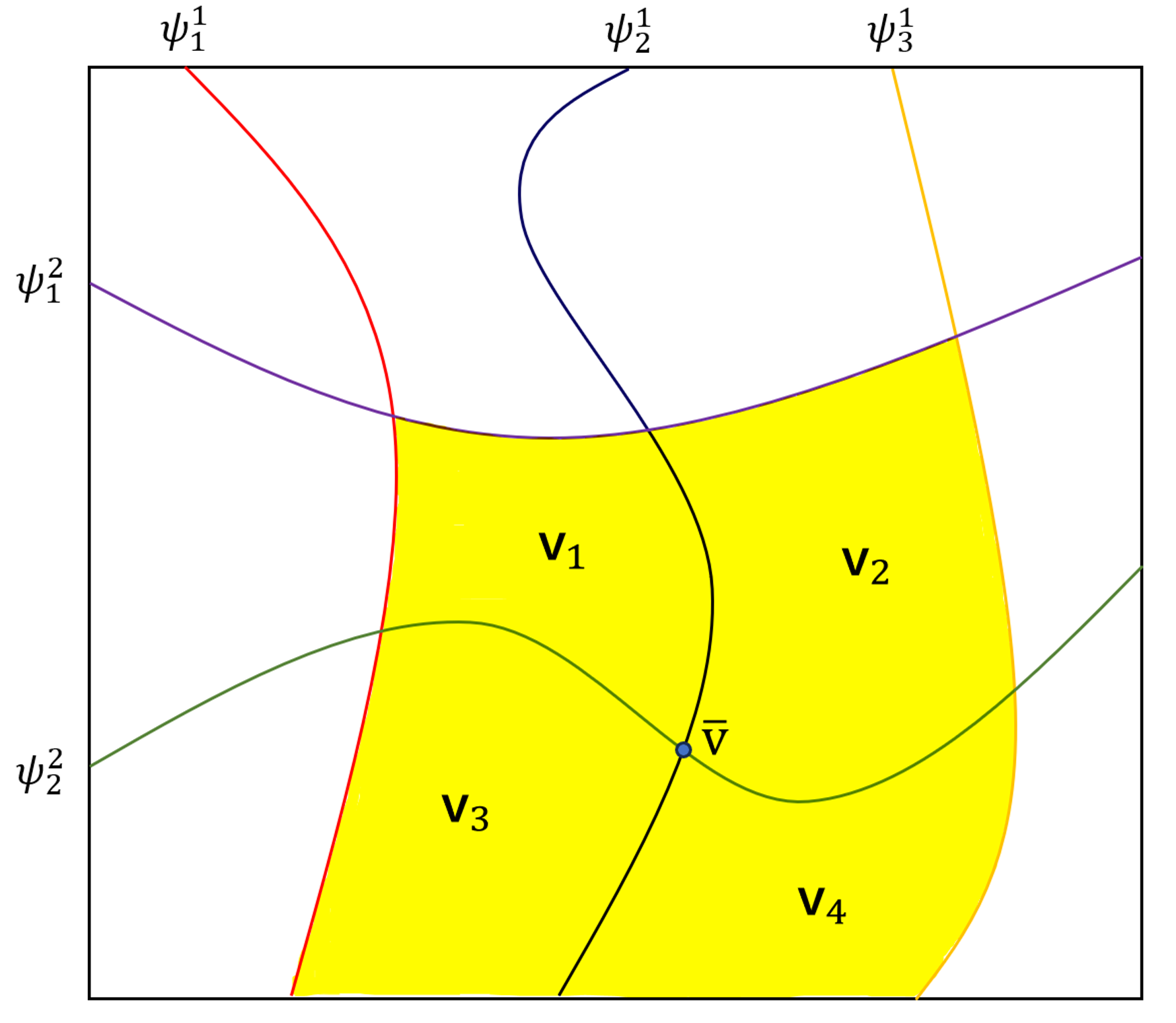}
  \caption{An intersection point can determine all the subregions connected to it.}
 \label{4}
\end{figure}
The proposed algorithm is depicted in Algorithm~\ref{alg1}.
\begin{algorithm}
 
\caption{PAT algorithm for Positive Semidefinite Quadratic Maximization} 
\label{alg1} 
\begin{algorithmic}[1] 
\REQUIRE Complex matrix $\mathbf{R}$, set $\Phi_n$, $n=1,2\cdots,N$. 
\ENSURE Optimal discrete phase configurations  $\mathbf{v}^*$ 
\STATE Generate the boundary point set $\Psi_n$ as in (17).
\REPEAT
\STATE Select intersection equations and solve them according to equation (22a) to obtain $\overline{\mathbf{v}}$.
\STATE Check if either $\overline{\mathbf{v}}$ or $-\overline{\mathbf{v}}$ satisfies equation (22b). 
\IF{Satisfying equation (22b)}
\STATE Recovering the corresponding $\mathbf{v}$ according to equation (18) and adding it to $V$.
\ENDIF
\UNTIL {All intersection points have been identified.}
\STATE Traverse $V$ to obtain the optimal solution $\mathbf{v}^*$.
\RETURN $\mathbf{v}^*$.
\end{algorithmic} 
\end{algorithm}

\section{Evolutionary Algorithm for Multi-user Scenarios}\label{Section4}

In this section, we extend the analysis to the multi-user scenario. Unlike the single-user case, Problem (P4) in the multi-user context cannot be fully reformulated as a quadratic problem, and its scale increases significantly, making the direct application of the PAT algorithm computationally prohibitive for obtaining the global optimal solution. To address this, we first perform an approximate transformation of the problem and derive a lower bound for the original problem. Building on this, we propose an evolutionary version of the PAT algorithm, referred to as the Efficient-Partition-and-Traversal (E-PAT) algorithm, which effectively resolves the aforementioned challenges. 

\subsection{Problem reformulation}

It is well known that for any $N\times N$ matrix $\mathbf{A}$, the maximum eigenvalue satisfies $\mu_{max}(\mathbf{A}) \ge \frac{1}{N}\mathrm{Tr}(A)$. Hence, we obtain
\begin{equation}
    \begin{aligned}
        &\mu_{max}(\sum_{m=1}^{M}(\prod_{i \ne m}^{M}\sigma_i^2)\mathbf{R}_m\mathbf{v}\mathbf{v}^\mathrm{H}\mathbf{R}_m^\mathrm{H}) \\
        &\ge \frac{1}{D}\mathrm{Tr}(\sum_{m=1}^{M}(\prod_{i \ne m}^{M}\sigma_i^2)\mathbf{R}_m\mathbf{v}\mathbf{v}^\mathrm{H}\mathbf{R}_m^\mathrm{H}) \\
        & = \frac{1}{D}\sum_{m=1}^{M}\mathrm{Tr}((\prod_{i \ne m}^{M}\sigma_i^2)\mathbf{R}_m\mathbf{v}\mathbf{v}^\mathrm{H}\mathbf{R}_m^\mathrm{H}) \\
        & = \frac{1}{D}\sum_{m=1}^{M}(\prod_{i \ne m}^{M}\sigma_i^2)\mathbf{v}^\mathrm{H}\mathbf{R}_m^\mathrm{H}\mathbf{R}_m\mathbf{v} \\
        & = \frac{1}{D}\mathbf{v}^\mathrm{H}(\sum_{m=1}^{M}(\prod_{i \ne m}^{M}\sigma_i^2)\mathbf{R}_m^\mathrm{H}\mathbf{R}_m)\mathbf{v}. \\
    \end{aligned}
\end{equation}
Let $\sum_{m=1}^{M}(\prod_{i \ne m}^{M}\sigma_i^2)\mathbf{R}_m^\mathrm{H}\mathbf{R}_m = \mathbf{T}^\mathrm{H}\mathbf{T}$, where $\mathbf{T} \in \mathbb{C}^{MD\times N}$, then (P4) can be simplified to
\begin{equation}
    \begin{aligned}
        &\max_{\mathbf{v}}\mathbf{v}^\mathrm{H}\mathbf{T}^\mathrm{H}\mathbf{T}\mathbf{v}\\
        &\mathrm{s.t. \ arg}\{v_n\}\in \Phi_n, \forall{n}.
    \end{aligned}
\end{equation} 
Through the series of transformations described above, we have converted the problem (P4) into a semidefinite quadratic maximization problem, thereby deriving an upper bound for the transmission power. Compared to the single-user case, the rank of $\mathbf{T}$ is higher than that of $\mathbf{R}$ in (P5). As a result, although the PAT algorithm can still achieve the global optimal solution, its computational complexity reaches $O(N^{2MD-1})$. Based on this observation, in the next subsection, we propose a E-PAT algorithm, an improvement upon the PAT algorithm, which strikes a balance between computational cost and performance.

\subsection{A Efficient-Partition-and-Traversal Algorithm For Positive Semidefinite Quadratic Maximization}

First, let us review the workflow of the PAT algorithm. For a quadratic problem with rank $MD$, the PAT algorithm introduces an auxiliary variable $\overline{\mathbf{v}}$, transforming the problem of solving $\mathbf{v}$ into solving $\overline{\mathbf{v}}$. It then establishes boundary equation systems using Equation (18) to partition the space where $\overline{\mathbf{v}}$ resides. By traversing the intersection points of these boundaries, the subregion related to the optimal solution can be identified, leading to the optimal solution. Thus, the key to the algorithm is finding the optimal subregion after partitioning. The PAT algorithm achieves this by finding the intersection points of $(2MD-1)$ boundary equations.

Thus, the computational overhead of the PAT algorithm primarily arises from solving equation systems. Reducing the number of equations in each system can significantly decrease the number of systems that need to be solved, thereby reducing the algorithm's computational cost. Suppose the number of equations per system is reduced to $d$, then each system can yield $(2MD-d)$ orthogonal points on the $(2MD-1-d)$-dimensional boundary. As shown in Figure 5, during the search process of the E-PAT algorithm, unlike the PAT algorithm, E-PAT can find multiple points on the boundary instead of just one intersection point. If one of these points is connected to the optimal subregion (e.g., the red point in the figure), E-PAT can locate the optimal solution.
\begin{figure}
    \centering
    \includegraphics[width=3.4in]{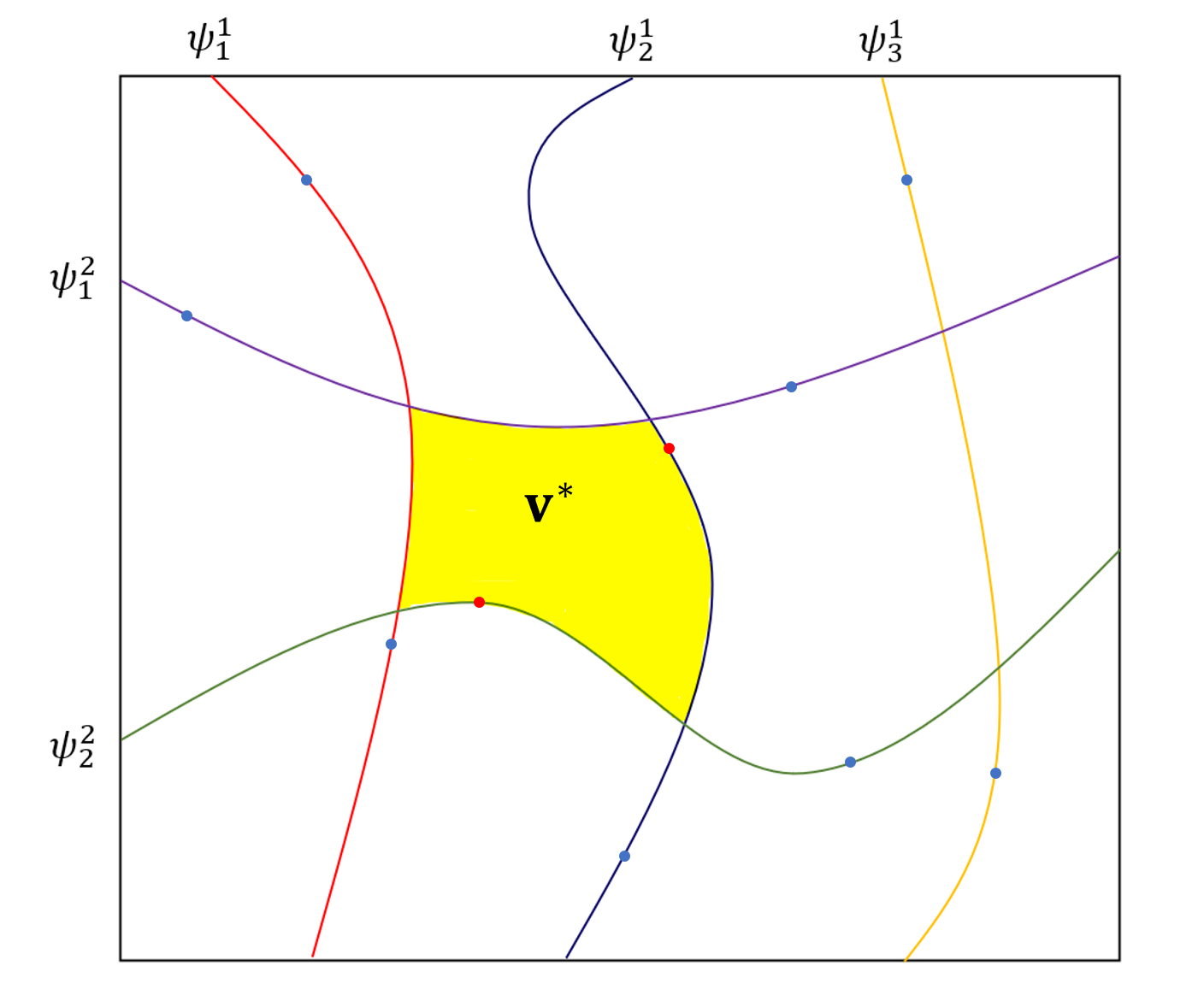}
    \caption{The search process of the E-PAT algorithm (with the red dot indicating the optimal solution).}
   \label{5}
  \end{figure}
\begin{figure}
    \centering
    \includegraphics[width=3.4in]{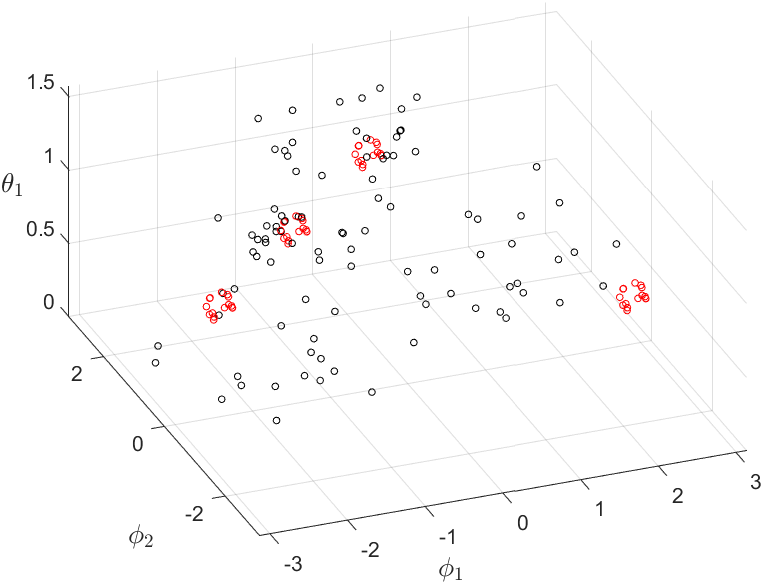}
    \caption{The search result of the E-PAT algorithm.}
   \label{6}
  \end{figure}
To intuitively demonstrate the feasibility of the E-PAT algorithm, we visualize search result of the E-PAT algorithm under the condition $MD=d=2$, as shown in Figure 6. The figure highlights the vertices of the optimal subregion (in red) and the points corresponding to the solutions provided by the algorithm (in black). It can be observed that the black points consistently approach the red points, indicating that the algorithm can find the optimal solution with a certain probability $p_c$, or provide a near-optimal solution. From Figure 5, it is evident that not every solution of the equation system is valid. Only when the boundary corresponding to the equation system is connected to the optimal subregion do these solutions hold significance. By combining the arithmetic mean-geometric mean inequality, we obtain 
\begin{equation}
    p_c \ge 1-p_e^B,
\end{equation}
where $p_e$ represents the average probability of failing to obtain the optimal solution in a single iteration of the equation system resolution process, and $B$ represents the number of $(2MD-1-d)$-dimensional boundaries connected to the optimal subregion. Since the simplex is the convex polytope with the smallest number of vertices, we can derive
\begin{equation}
    B \ge \left(\begin{matrix}2MD\\2MD-d\\\end{matrix}\right),
\end{equation}
where $\left(\begin{matrix}N\\n\\\end{matrix}\right)$ denotes the binomial coefficient, defined as the number of ways to choose $n$ elements from a set of $N$ elements.
\begin{Theorem}
    An $N$-dimensional simplex has at least $\left(\begin{matrix}N+1\\n+1\\\end{matrix}\right)$ $n$-dimensional faces.
\end{Theorem}
\begin{proof}
    An $N$-dimensional simplex has at least $N+1$ vertices. To construct an $n$-dimensional face, we need to choose $n+1$ vertices from these $N+1$ vertices. Consequently, there are at least $\left(\begin{matrix}N+1\\n+1\\\end{matrix}\right)$ $n$-dimensional faces.
\end{proof}
The value of $p_e$ is difficult to analyze precisely. However, by solving the boundary equations, we obtain $(2MD-d)$ points and generate more points based on some strategy, thereby increasing the probability of finding the optimal solution. Investigating which strategies can achieve the lowest $p_e$ with the least complexity is a key direction for future research. Given that the focus here is on reducing algorithm complexity, we do not employ any specific strategy at this stage. Consequently, the lower bound for $p_c$ is
\begin{equation}
    p_{c,low} = 1-p_e^{\left(\begin{matrix}2MD\\2MD-d\\\end{matrix}\right)}.
\end{equation}

\section{Numerical Simulation And Analysis}\label{Section5}
\begin{figure}
    \centering
    \includegraphics[width=3.4in]{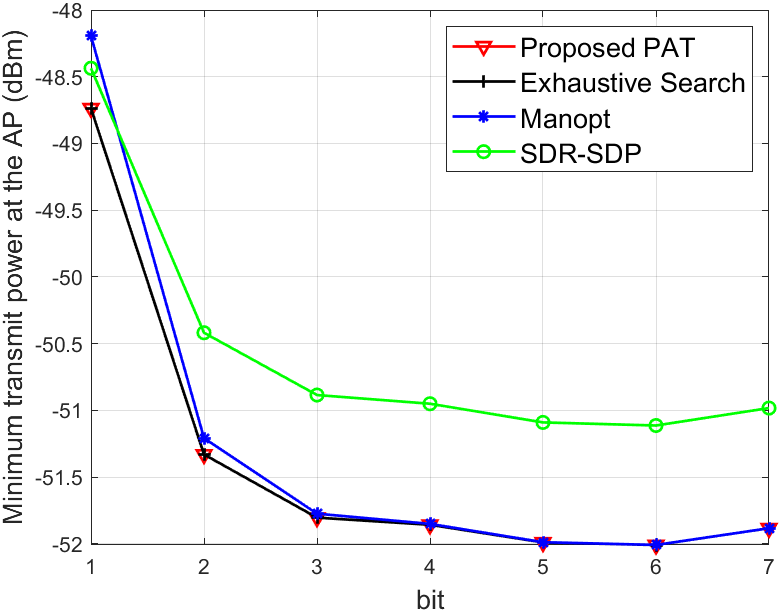}
    \caption{A comparison of minimum transmit power performance at different bit counts.}
   \label{7}
\end{figure}
\begin{figure}
    \centering
    \includegraphics[width=3.4in]{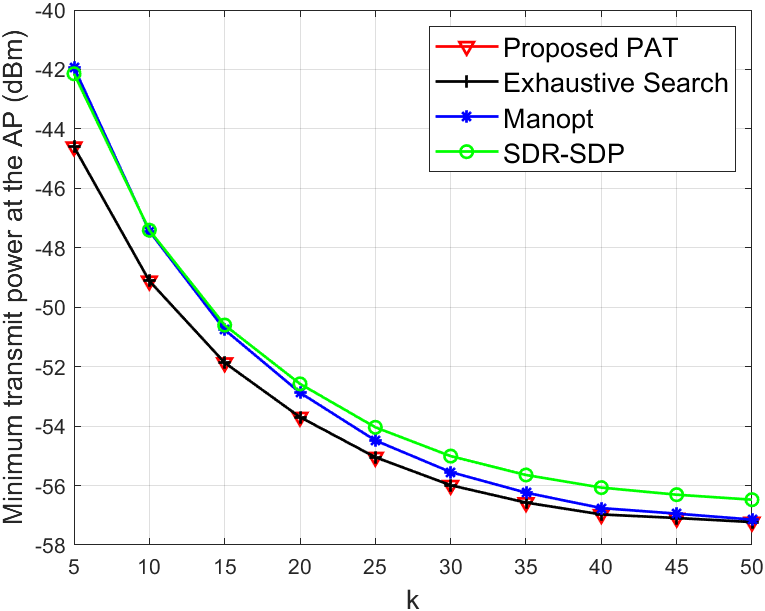}
    \caption{A comparison of minimum transmit power performance as function $k$.}
   \label{8}
  \end{figure}
In this section, we begin by analyzing the impact of non-uniform phase quantization on RIS performance through numerical simulations. Subsequently, we compare the proposed algorithm's performance with advanced methods (e.g., SDR-SDP and Manopt) and with exhaustive search, and conduct a complexity analysis of the algorithm. Notably, the continuous solutions obtained through SDR-SDP and Manopt methods are transformed into corresponding discrete solutions using the CPP method. Finally, based on numerical simulations, we provide strategies for selecting a proper parameter $d$ in the E-PAT algorithm.

The channel model parameter is based on independent and identically distributed (i.i.d.) Gaussian channels with zero mean and variance $\sigma_0^2$, i.e., $\mathbf{h}_{k,m}^\mathrm{H}, \mathbf{G}_k\sim \mathcal{CN}(0,\sigma_0^2)$. The lower bound of the user’s average SNR $\gamma$ is set to $40$ dBm and the background noise power level $\sigma_k^2$ is set to $-50$ dBm. Additionally,  we constrain RIS units to be either $1$-bit or $2$-bit, and randomly distribute $\Phi_n$ within $[0,2\pi)$.
\begin{figure}
    \centering
    \includegraphics[width=3.4in]{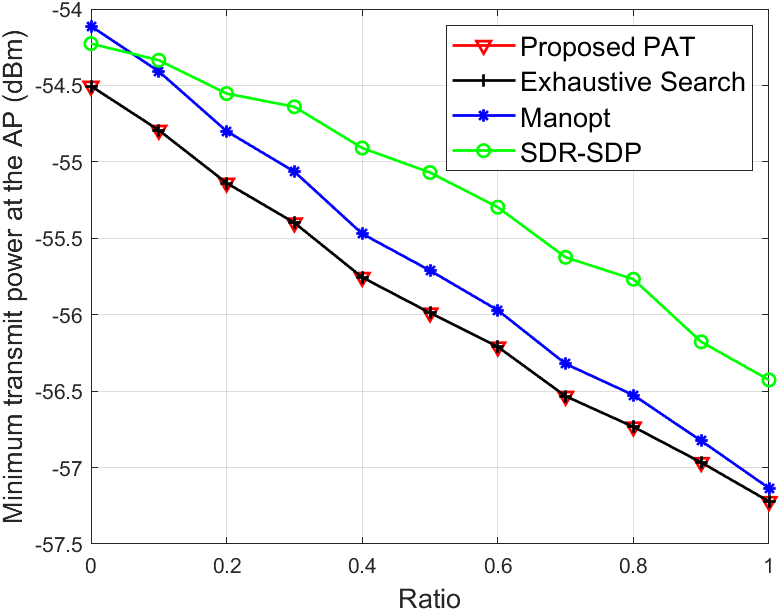}
    \caption{A comparison of transmit power performance with different ratios of 2-bit units among all units.}
   \label{9}
\end{figure}
\subsection{Analysis of the Impact of Non-Uniform Phase Quantization}

\subsubsection{The Impact of Phase Resolution}
We designed a set of experiments to compare the performance of traditional algorithms under varying phase resolutions, with the results presented in Fig~\ref{7}. When the phase resolution increases from 1-bit to 7-bit, the Manopt algorithm achieves near-optimal solutions at 4-bit. However, for 1-bit or 2-bit cases, the Manopt and SDR algorithms exhibit a deviation of approximately 1 dB compared to the proposed algorithm. This highlights the inherent performance loss associated with the mainstream approach of discretizing solutions obtained from continuous algorithms to generate discrete outputs.

From a theoretical perspective, the observed performance losses can be attributed to two primary factors. First, the discretization introduced during the optimization process alters the continuous gradient information of the objective function and constraints, leading to a degradation in the algorithm's convergence rate and overall efficiency. Second, the discretization of the solution space can result in the convergence of continuous algorithms to suboptimal, or potentially infeasible, solutions that do not adequately reflect the underlying problem structure. 

\subsubsection{Non-Uniformity at Unit Level}
We conduct an experiment to further analyze the impact of non-uniform quantization of phases. In the experimental setup, all units are configured with $2$-bit resolution, and the phase distribution is given by $\{0, \frac{k\pi}{20}, \frac{k\pi}{10}, \frac{3k\pi}{20}\}$. In Fig~\ref{8}, we compare the results of various algorithms for $k = 1, 2, \cdots, 10,$ where $k$ determines the gap between the four phases. As $k$ decreases, the phase distribution becomes more uneven. The results show that regardless of the value of $k$, the proposed algorithm achieves the same transmit power performance as the exhaustive search. When $k$ is large, meaning the phase distribution is nearly uniform, the Manopt algorithm performs similarly to the proposed algorithm. However, when $k$ is small, for instance, when $k=1$, the performance gap between the existing algorithms and the proposed algorithm exceeds 2 dB. This highlights the significant advantage of the proposed algorithm in handling unit-level phase non-uniformity compared to traditional algorithms.
\begin{figure}
    \centering
    \subfigure[]{
    \label{10-1}
    \includegraphics[width=0.48\linewidth]{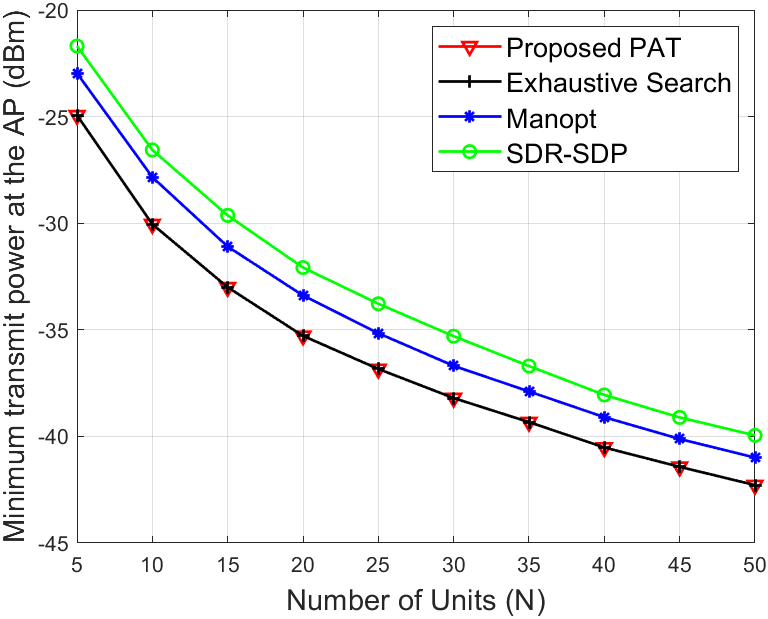}}
    \subfigure[]{
    \label{10-2}
    \includegraphics[width=0.48\linewidth]{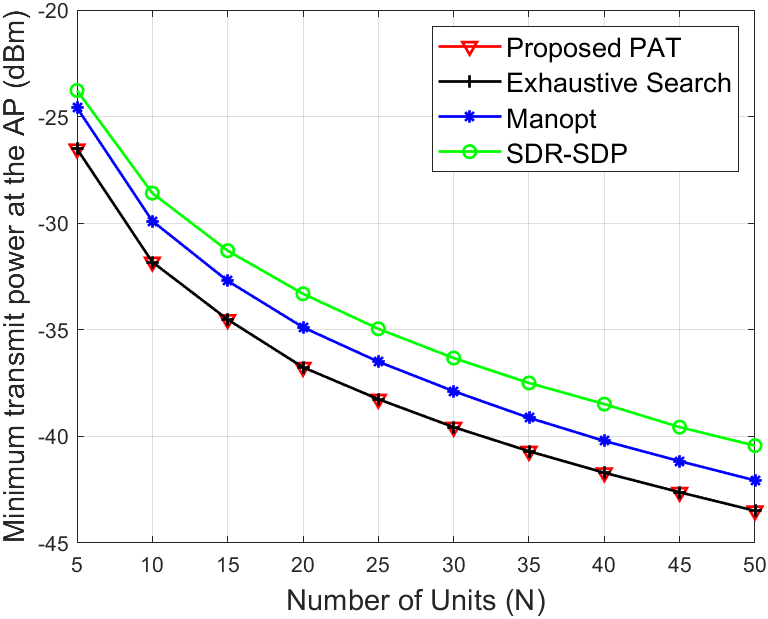}}
    \subfigure[]{
    \label{10-3}
    \includegraphics[width=0.48\linewidth]{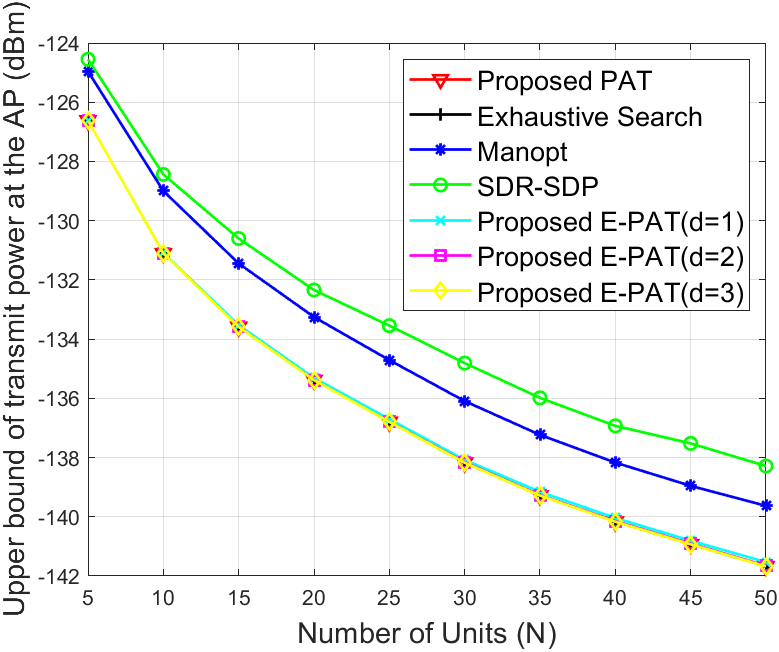}}
    \subfigure[]{
    \label{10-4}
    \includegraphics[width=0.48\linewidth]{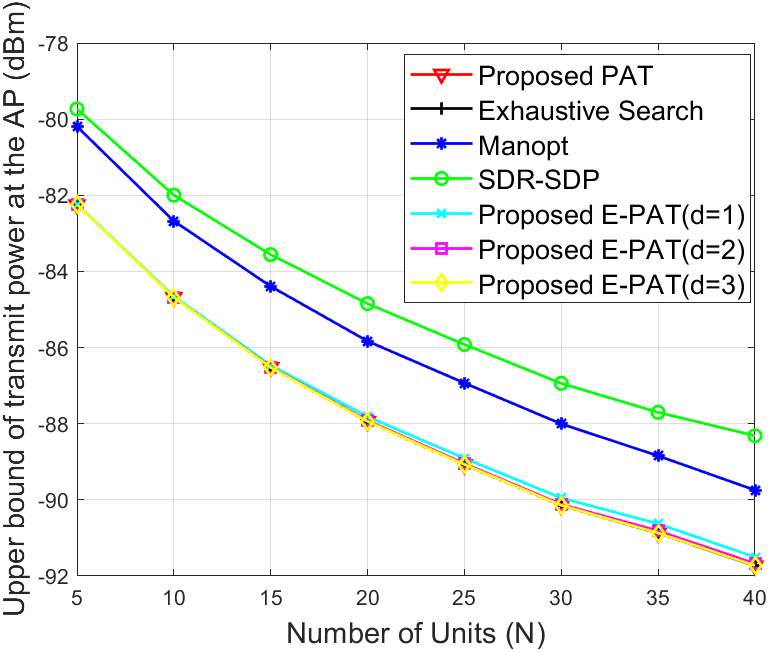}}
    \caption{A comparison of minimum transmit power performance as function $N$. (a) $M = 1, D = 1$. (b) $M = 1, D = 2$. (c) $M = 3, D =1$. (d) $M = 2, D =2$}
    \label{10}
\end{figure}
\begin{figure*}[b]
    \centering
    \subfigure[]{
    \label{11-1}
    \includegraphics[width=0.32\linewidth]{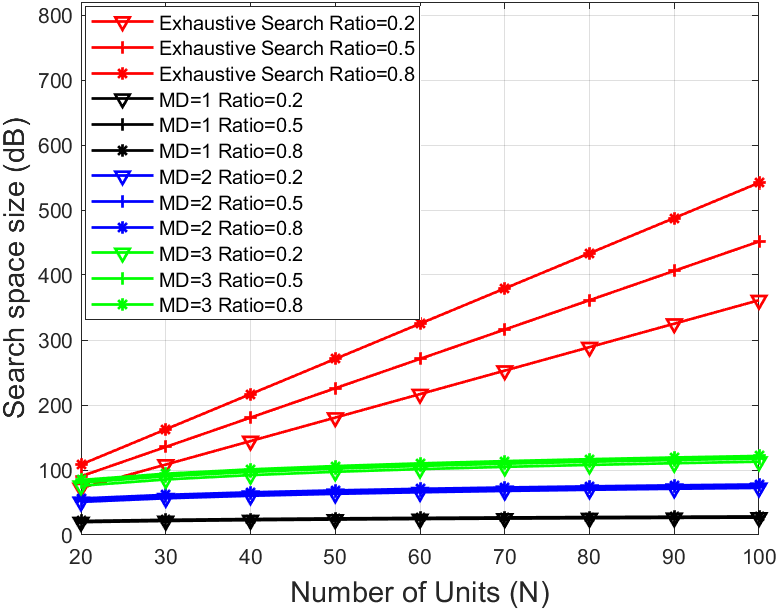}}
    \hfill
    \subfigure[]{
    \label{11-2}
    \includegraphics[width=0.32\linewidth]{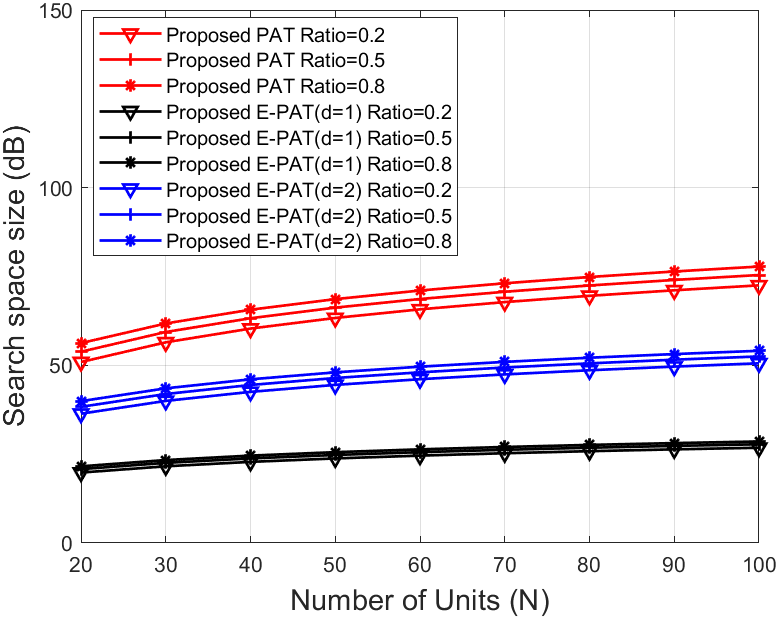}}
    \hfill
    \subfigure[]{
    \label{11-3}
    \includegraphics[width=0.32\linewidth]{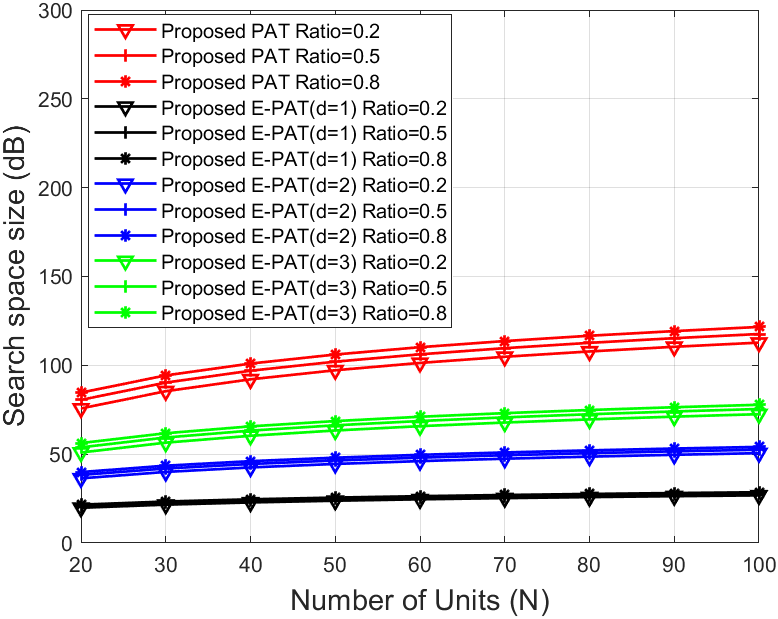}}
    \caption{A comparison of search space size as function $N$. (a) Exhaustive research versus PAT. (b) $MD=2$, E-PAT versus PAT. (c) $MD=3$, E-PAT versus PAT.}
    \label{11}
\end{figure*}
\subsubsection{Non-Uniformity at Array Level}
To analyze the impact of non-uniform bit resolution, we compare the results of various algorithms under different ratios of 2-bit units among all units, as shown in Fig~\ref{9}. The findings demonstrate that the proposed algorithm consistently maintains the same transmission power performance as the exhaustive search method. As the ratio increases and the bit distribution becomes non-uniform, the gap between SDR-SDP and the proposed algorithm gradually widens. However, as the ratio continues to increase and the bit distribution becomes more uniform, the gap between the two methods diminishes. This indicates that the proposed algorithm demonstrates a clear advantage over traditional algorithms when dealing with phase non-uniformity at the array level.

\subsection{Comprehensive Analysis of the Proposed Algorithm}

\subsubsection{Minimization of Transmit Power}
In Fig~\ref{10}, we compare the transmit power performance of various algorithms under different values of $M$ and $D$.  Here, $M$ represents the number of users, $D$ denotes the number of transmitting antennas, and $MD$ indicates the rank of the problem being solved. From subfigures (a) and (b), it is evident that the PAT algorithm successfully achieves the global optimal solution, outperforming both the SDR-SDP and Manopt algorithms. Specifically, the PAT algorithm outperform both SDR-SDP and Manopt by at least 2 dB and 4 dB, respectively. From subfigures (c) and (d), it can be observed that as the value of $d$ increases, the performance gap between the E-PAT algorithm and the PAT algorithm gradually narrows. Additionally, when $d=1$, the E-PAT algorithm exhibits a minor performance gap compared to the exhaustive search, but shows a larger gap of approximately 2 dB compared to the Manopt and SDR-SDP algorithm.

\subsubsection{Search Complexity Analysis}

Let $\mathcal{I}$ represent the set containing all possible $I$, and the size of the set $V$ obtained by the proposed algorithm is given by $\sum_{I \in \mathcal{I}}2^{L}\prod_{n \in I}b_n$, where $L$ is the size of the set $I$. Note that, the size of the search domain for the exhaustive search is $\prod_{n=1}^N b_n$. We compare the search space sizes of the PAT algorithm, the E-PAT algorithm, and the exhaustive method, with the proportion of 2-bit units set to 0.2, 0.5, and 0.8 across all units. As shown in Fig~\ref{11-1}, the gap in search space size between the PAT algorithm and the exhaustive method grows significantly, as $N$ increases. The PAT algorithm's search space is significantly smaller than that of exhaustive search, being over 200 dB smaller when $N=100$. Fig~\ref{11-2} and (c) illustrate the comparison of search space sizes between the PAT and E-PAT algorithms when $MD = 2$ and $MD = 3$, respectively. It is evident that, as the parameter $d$ increases, the search space of the E-PAT algorithm also increases. However, it remains significantly smaller than that of the PAT algorithm. For example, when $N = 100$, the search space of the E-PAT algorithm is at least 50 dB smaller than that of the PAT algorithm under different ratios,. In addition to the algorithm's advantage in terms of the size of the search space, the independent computation of subregions allows for parallel operation. Moreover, the necessity to store only the optimal solution, without the need for non-optimal ones, significantly reduces the memory requirements. These advantages make our proposed PAT algorithm still particularly applicable in large-scale scenarios.

\subsubsection{Optimality of the E-PAT Algorithm}
To further evaluate the performance of the E-PAT algorithm, we conducted a series of repeated experiments and derived relevant performance metrics, as shown in Table~\ref{Performance}. The experimental results show that both the relative error and optimality probability of the E-PAT algorithm decrease as the parameter $d$ increases, while its computational complexity increases accordingly. Specifically, its relative error consistently remains below the $10^{-2}$ level, while its complexity is less than $1\%$ of that of the PAT algorithm. Thus, by introducing a small performance loss, the E-PAT algorithm successfully achieves a significantly lower computational complexity compared to the PAT algorithm. Moreover, as we discussed in the Search Complexity Analysis, the search space of the E-PAT algorithm scales with O(N) when $d=1$. Coupled with the strategy we previously mentioned for achieving lower $p_e$, the E-PAT algorithm enables optimal solution retrieval with low computational complexity. This positions the E-PAT algorithm as a more competitive and practically valuable solution for large-scale applications.
\begin{table*}[ht]
    \centering
    \caption{Performance of E-PAT Algorithm with Varying Parameters Across Multiple Repeated Experiments}
    \label{Performance}
    \setlength{\tabcolsep}{2.5mm} 
    \begin{tabular}{cccccccccc}
    \toprule
    \multirow{2}{*}{N} & \multirow{2}{*}{MD} & \multirow{2}{*}{d} & \multirow{2}{*}{Relative Error}   & \multicolumn{2}{c}{Complexity Ratio} & \multirow{2}{*}{Optimality Probability (\%)}\\
                        &&&& vs. PAT (\%) & vs. Exhaustive search (\%) & \\
    \midrule
    20 & 2 & 1 & 6.21E-01 & 5.70E-02 & 4.17E-05 & 33.70 \\
    30 & 2 & 1 & 7.80E-01 & 2.28E-02 & 1.96E-09 & 21.26 \\
    40 & 2 & 1 & 9.90E-01 & 1.13E-02 & 2.08E-14 & 5.02 \\
    20 & 2 & 2 & 1.28E-02 & 2.81E+00 & 6.33E-04 & 92.94 \\
    30 & 2 & 2 & 3.38E-02 & 1.76E+00 & 2.32E-08 & 82.94 \\
    40 & 2 & 2 & 3.46E-02 & 1.50E+00 & 2.42E-10 & 80.54 \\
    20 & 3 & 1 & 1.35E+00 & 1.40E-04 & 4.17E-05 & 14.12 \\
    30 & 3 & 1 & 1.50E+00 & 7.28E-05 & 5.77E-06 & 11.38 \\
    40 & 3 & 1 & 1.91E+00 & 2.16E-05 & 7.46E-09 & 5.64 \\
    20 & 3 & 2 & 1.26E-01 & 7.40E-03 & 2.20E-03 & 64.34 \\
    30 & 3 & 2 & 1.63E-01 & 4.40E-03 & 3.46E-04 & 58.78 \\
    40 & 3 & 2 & 2.76E-01 & 1.70E-03 & 5.96E-07 & 39.28 \\
    20 & 3 & 3 & 1.80E-03 & 2.45E-01 & 7.31E-02 & 97.90 \\
    30 & 3 & 3 & 2.60E-03 & 1.68E-01 & 1.33E-02 & 97.48 \\
    40 & 3 & 3 & 4.60E-03 & 8.97E-02 & 3.09E-05 & 95.44 \\
    \bottomrule
    \end{tabular}
\end{table*}

\subsection{Selection Strategy for $d$ in the E-PAT Algorithm}
The term $d$ directly influences the performance of the E-PAT algorithm. Based on equation (28), it can be inferred that the optimal value of $d$ is $MD$. This is because, as $d$ increases, $\left(\begin{matrix}2MD\\2MD-d\\\end{matrix}\right)$ initially increases and then decreases, reaching its maximum when $d=MD$, while $p_e$ continuously decreases. Consequently, $p_c$ rapidly increases with $d$ and the growth rate slows down when $d>MD$. Simultaneously, the complexity ratio increases sharply with $d$. Therefore, considering the balance between relative error and complexity ratio, we conclude that the optimal choice for $d$ is when $d=MD$. For reference,  we present the empirical trade-off between relative error and complexity ratio for various values of $d$ in the E-PAT algorithm, as shown in Fig~\ref{12}. 
The experimental results show that when $MD=3$ and $d=3$, the relative error rate is almost zero. However, when $d$ is reduced to 2, the relative error rate increases significantly. Similarly, when $MD=4$ the relative error rate increases substantially when $d$ decreases to 3. Therefore, based on the experimental results, it can be concluded that selecting $d=MD$ as the optimal value is correct.
\begin{figure}[htbp]
    \centering
    \subfigure[]{
    \label{12-1}
    \includegraphics[width=1\linewidth]{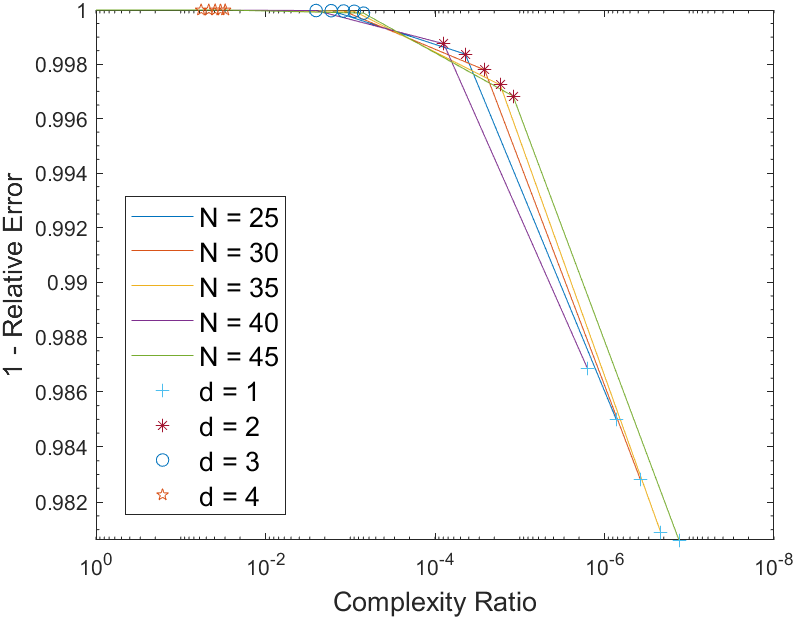}}
    \subfigure[]{
    \label{12-2}
    \includegraphics[width=1\linewidth]{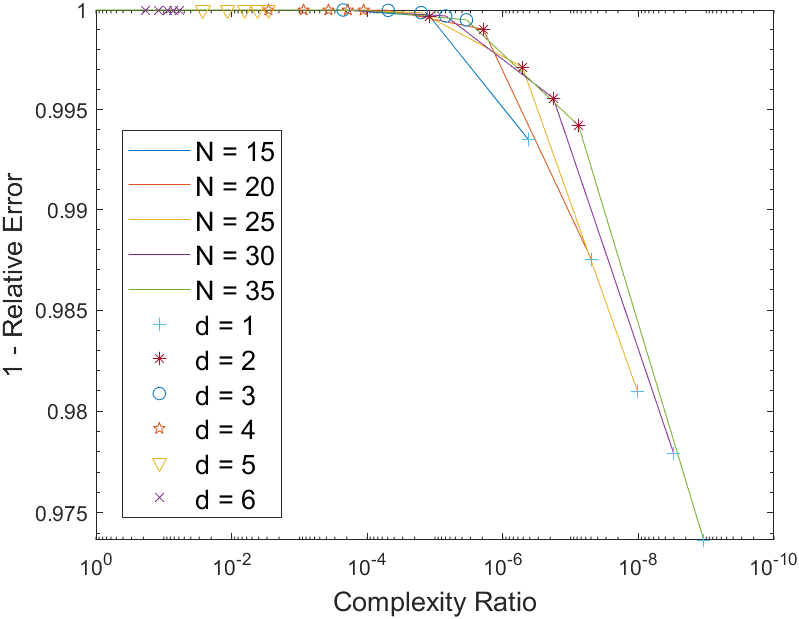}}
    \caption{The trade-off between relative error and complexity ratio in the E-PAT algorithm. (a) $MD=3.$ (b) $MD=4.$}
    \label{12}
\end{figure}

\section{CONCLUSION}\label{Section6}


In this paper, we investigate the beamforming optimization problem for RIS-assisted communication systems under non-uniform phase quantization. To address the challenges posed by non-uniform unit-level and array-level phase configurations, we formulate a discrete optimization problem aimed at minimizing the total transmit power while meeting a minimum SNR requirement. We propose a Partition-and-Traversal (PAT) algorithm that guarantees the global optimum by systematically partitioning the search space and exhaustively traversing it. To further balance computational complexity and performance, we introduce the Efficient Partition-and-Traversal (E-PAT) algorithm, which significantly reduces computational overhead while maintaining near-optimal performance. Extensive numerical simulations validate that the PAT algorithm consistently attains the global optimum, outperforming conventional methods, while the E-PAT algorithm offers a practical solution for large-scale RIS-assisted systems. Additionally, we provide an analysis of the optimal selection strategy for the E-PAT algorithm’s parameter $d$, ensuring a favorable trade-off between computational efficiency and solution accuracy.

An important research direction is to further refine the E-PAT algorithm by developing strategies to reduce the probability of missing the optimal solution while keeping computational complexity low. Additionally, extending the proposed optimization framework to more general scenarios, such as wideband multi-carrier communications and RIS-aided MIMO systems, presents promising opportunities for future work.
\bibliographystyle{IEEEtran}
\bibliography{Reference}
%
%
%
%
%
%
%
%
%
%

\newpage

 


\vspace{11pt}


\vfill
\end{document}